\documentclass[11pt]{article}

\usepackage{amsfonts}
\usepackage{amsmath}
\usepackage{framed}
\usepackage{color}
\usepackage{graphicx}

\usepackage[letterpaper, top=1in, bottom=1in, left=1in, right=1in]{geometry}

\def\ie{{\itshape{i.e.}}}
\def\eg{{\itshape{e.g.}}}

\def\etal{{\itshape et al.}}

\usepackage{amsthm}
\theoremstyle{plain}
\newtheorem{theorem}{Theorem}

\newtheorem{lemma}{Lemma}
\newtheorem{conjecture}{Conjecture}
\newtheorem{proposition}{Proposition}
\theoremstyle{definition}
\newtheorem{example}{Example}

\newcommand{\BIGO}[1]{{\mathcal{O}{\kern-2pt}}\left( #1 \right)}
\newcommand{\BIGOMEGA}[1]{{\Omega{\kern-2pt}}\left( #1 \right)}
\newcommand{\LITTLEO}[1]{{\mathcal{o}{\kern-2pt}}\left( #1 \right)}
\newcommand{\LITTLEOMEGA}[1]{{\omega{\kern-2pt}}\left( #1 \right)}
\newcommand{\BIGTHETA}[1]{{\Theta{\kern-2pt}}\left( #1 \right)}

\newcommand{\ceil}[1]{\left\lceil #1 \right\rceil}

\def\isdef{\buildrel {\mathrm{def}} \over =}

\def\RR{{\mathbb{R}}}

\def\Cnd{{\mathcal{C}_{n,d}}}
\def\SND{{\mathcal{S}_{n,d}}}

\newcommand{\PROB}[1]{{\mathbb{P}{\kern-2pt}}\left\{ #1 \right\}}
\newcommand{\EXP}[1]{{\mathbb{E}{\kern-2pt}}\left\{ #1 \right\}}
\newcommand{\VAR}[1]{{\mathbf{Var}{\kern-2pt}}\left\{ #1 \right\}}
\newcommand{\STDEV}[1]{{\mathrm{stdev}{\kern-2pt}}\left\{ #1 \right\}}
\newcommand{\tendsinlaw}{\mathop{\lower2pt\hbox{$\overset{\kern-1pt\lower2pt\hbox{$\scriptstyle\mathcal{L}$}}{\rightarrow}$}}}
\newcommand{\inlaw}{\mathop{\lower0pt\hbox{$\overset{\kern-1pt\lower0pt\hbox{$\scriptstyle\mathcal{L}$}}{=}$}}}

\newcommand{\BERNOULLI}[1]{{\mathrm{Bernoulli}{\kern-2pt}}\left( #1 \right)}
\newcommand{\BINOMIAL}[1]{{\mathrm{binomial}{\kern-2pt}}\left( #1 \right)}
\newcommand{\BETA}[1]{{\mathrm{beta}{\kern-2pt}}\left( #1 \right)}

\newcommand{\hide}[1]{}

\begin{document}

\title{Random hyperplane search trees in high dimensions}
\author{Luc Devroye\footnotemark[2] \\ \tt{\small luc.devroye@gmail.com} \and James King\footnotemark[3] \\ \tt{\small jamie.king@gmail.com}}
\date{\today}
\maketitle

\renewcommand{\thefootnote}{\fnsymbol{footnote}}

\footnotetext[2]{School of Computer Science,
McGill University, Montreal, Canada H3A 2A7}
\footnotetext[3]{Department of Physics, University of Oxford, Clarendon Laboratory, Parks Road, Oxford, UK OX1 3PU}
\renewcommand{\thefootnote}{\arabic{footnote}}

\hide{
\title{%
  {Random hyperplane search trees in high dimensions}%
  \thanks{This research was supported by NSERC grant A3456.}
}

\author{%
  Luc~Devroye%
  \thanks{\affil{McGill University}, 
          \email{luc.devroye@gmail.com}}\,
  and James King%
  \thanks{\affil{University of Oxford}, 
          \email{jamie.king@gmail.com}}
}
%
}


\begin{abstract}
Given a set $S$ of $n\geq d$ points in general position in $\RR^d$, a \emph{random hyperplane split} is obtained by sampling $d$ points uniformly at random without replacement from $S$ and splitting based on their affine hull.  A \emph{random hyperplane search tree} is a binary space partition tree obtained by recursive application of random hyperplane splits.  We investigate the structural distributions of such random trees with a particular focus on the growth with $d$.  A \emph{blessing of dimensionality} arises---as $d$ increases, random hyperplane splits more closely resemble perfectly balanced splits; in turn, random hyperplane search trees more closely resemble perfectly balanced binary search trees.

We prove
that, for any fixed dimension $d$, a random hyperplane search tree storing $n$ points has height at most $(1+\mathcal{O}(1/\sqrt d))\log_2{n}$ and average element depth at most $(1+\mathcal{O}(1/d))\log_2{n}$ with high probability as $n\to\infty$.  Further, we show that these bounds are asymptotically optimal with respect to $d$.
\end{abstract}



\pagestyle{plain}


\section{Introduction}

Point sets in $\RR^d$ can be partitioned recursively by a number of possible trees.  The early, and still most popular, choices are the $k$-d tree and the quadtree.  The $k$-d tree takes a point from the set and partitions the space into two sets with a hyperplane containing the point that is perpendicular to one of the axes.  In a quadtree, the split is into $2^d$ quadrants obtained by shifting the origin to the point in question.  A lot of ink has been spilled on the analysis of the shapes of the trees for random point sets---for a summary and mini-survey, see Devroye \cite{devroye99}.  

In this paper, we focus on deterministic point sets, outside the control of the user, and random partitions that are built on them.  For example, in either of the two trees mentioned above, one could choose a splitting point uniformly at random, and make independent choices recursively on the subsets.  We assume throughout that points are in general position (no three on a line, no four on a plane, and so forth).  In both examples, if the set of data points lies on the moment curve $\big\{ (x, x^2, \ldots, x^d ): x \in \RR \big\}$, then the tree thus obtained is statistically equivalent to a random binary search tree.


Analysis of random tree data structures typically focusses on two functions that quantify the level of balance: the depth (specifically, the mean point depth) and the height (\ie, the maximum point depth).  These values are of particular practical importance---when searching for point in the tree, they correspond respectively to the average-case and worst-case query times.  For a perfectly balanced binary search tree, using $D^*_n$ and $H^*_n$ to denote the depth and the height, we have 
$$
\lim_{n\to\infty} \frac{D^*_n}{\log_2 n} ~=~ \lim_{n\to\infty} \frac{H^*_n}{\log_2 n} ~=~ 1~,
$$
which is the best we can hope for when dealing with binary trees.  

A random $k$-d tree in any dimension has the same shape as a random binary search tree---notably, the distribution does not depend on the structure of the point set, only its size.  We use $H_n$ and $D_n$ to denote the height and depth of a tree storing $n$ points.  It is known that 
$D_n / \log_2 n \to 2/\log 2 = 2.88539\dots$
 in probability \cite{lynch65,knuth73,devroye88}.  The limit law for $D_n$ was derived by Devroye \cite{devroye88}: $(D_n - 2 \log n )/\sqrt{2 \log n} ~ \tendsinlaw ~ {\cal N} (0,1)$.  Robson \cite{robson79}, Pittel \cite{pittel84}, and Devroye \cite{devroye86,devroye87} showed that 
$H_n / \log_2 n \to 4.31107/\log 2 = 6.21956\dots$ 
in probability.  See Mahmoud \cite{mahmoud92} for more background.

\subsection{Random hyperplane search trees}
For a given set $S$ of $n$ points in general position in $\RR^d$, a \emph{random hyperplane search tree} is constructed as follows.  If $n \geq d$, it selects at the root level $d$ points uniformly at random without replacement from the $n$ data points, and considers the hyperplane through these points, \ie, their affine hull.  These $d$ pivot points are associated with the root node and remain there.  The hyperplane splits the remaining $n-d$ points into two sets that are handled recursively and independently.  If $n < d$, no splitting is applied, and all $n$ points are associated with the root, which becomes a leaf.  This construction guarantees that each internal node holds $d$ data points and each leaf node holds between $0$ and $d-1$ data points.  See Figure \ref{fig:hst} for an example.

\begin{figure}
\centering
\includegraphics[width=\textwidth]{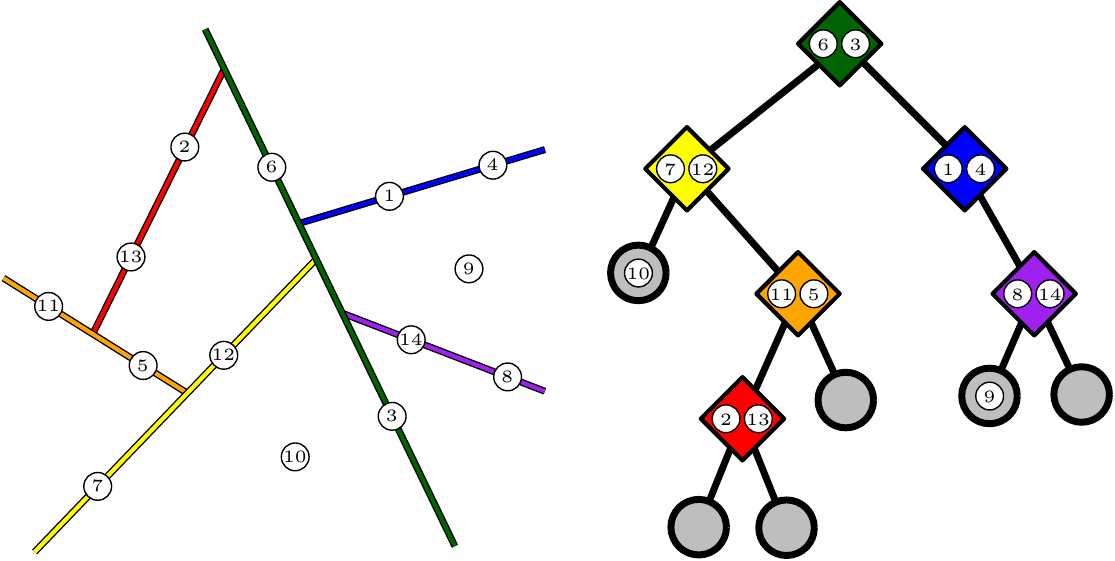}
\caption{\label{fig:hst}A hyperplane search tree in $\RR^2$: the point set and the hyperplane splits (left) and the corresponding tree data structure (right).  Internal tree nodes correspond to hyperplane splits and contain $d$ data points each.  External tree nodes correspond to cells and contain between 0 and $d-1$ data points each.}
\end{figure}

A key feature of hyperplane search trees is that they are constructed independently of the axes and are therefore robust to affine transformations of the underlying point set.  If the set of points contained by a $k$-d tree undergoes a rotation, the $k$-d tree would have to be reconstructed, however this is not the case for hyperplane search trees.

\paragraph{Applications}
Hyperplane search trees have been used since the 1970s in many applications of statistics.  For example, Mizoguchi \etal~\cite{mizoguchi77} highlighted their use in pattern recognition and You and Fu \cite{you76} considered their use as \emph{tree classifiers}.   Tree-based decisions in pattern recognition are popular because they take small computational efforts in terms of $n$. This is
especially crucial when decisions must be made on-line, in real-time.  Not only does the logarithmic behavior in $n$ matter, but also the asymptotic constants.  For an introduction to tree classification, see chapter 20 of Devroye, Gy\"orfi and Lugosi \cite{devroye96}.  In computational geometry, trees based upon partitions
of space by means of hyperplanes are ubiquitous.  See for example the
survey of Edelsbrunner and Van Leeuwen \cite{edelsbrunner83}, or the work of Haussler and Welzl \cite{haussler87} on simplex range queries.  For more examples and references, see section 2 of our previous paper \cite{devroye09}.

\hide{
Fuchs, Kedem and Naylor \cite{fuchs80} introduce the \underbar{{\sc bsp
trees}} (``binary space partition trees'') for use in graphics
applications. The space is split in two linear halfspaces; each
halfspace may in turn be split by a linear hyperplane, and so
forth. If a viewer sits in a given polyhedral set in this partition,
and wants to project the world onto his/her view plane, the {\sc bsp}
tree aids in establishing the order in which the polyhedral cells must be
drawn so as not to cause visibility problems. Basically, one should
consider polyhedra in depth-first-search order, where the
depth-first-search first visits halfspaces that would not contain the
viewer, so that polyhedra are visited from ``far'' to ``near'' (this
is called the painter's algorithm). For more on the hidden surface
elimination with the aid of {\sc bsp} trees, see Samet \cite{samet90a,samet90b}, or
Fuchs, Abram and Grant \cite{fuchs83}.  While {\sc bsp} trees are not
hyperplane trees (because we do not take data points to generate the
partition), they are intimately related and indicate interesting
applications of hyperplane trees in hidden surface elimination and
beam tracing. See also Sung and Shirley \cite{sung92} and Kaplan \cite{kaplan85}.
}

\subsection{Results}
Define $\SND = \{S:S\subset\RR^d,\, |S|=n,\, S~\text{is in general linear position}\}$.  For a set $S\in\SND$ we use $H(S)$ and $D(S)$ to denote the height and mean data point depth of a random hyperplane search tree built on $S$.  $H(S)$ and $D(S)$ are random variables.  By a trivial coupling argument we have that $H(S)$ stochastically dominates $D(S)$, \ie, $\PROB{D(S)\leq t} \geq \PROB{H(S)\leq t}$ for any value of $t$.  In order to more cleanly express bounds on $H(S)$ and $D(S)$, we define
\begin{align}
\label{def:ch}C_H(d) &\isdef \inf\left\{c\in\RR\,:\, \lim_{n\to\infty} \, \max_{S\in\SND} \PROB{\frac{H(S)}{\log_2 n} \,\leq\, c} ~ = ~ 1\right\},\\
\label{def:cd}C_D(d) &\isdef \inf\left\{c\in\RR\,:\, \lim_{n\to\infty} \, \max_{S\in\SND} \PROB{\frac{D(S)}{\log_2 n} \,\leq\, c} ~ = ~ 1\right\}.
\end{align}
The phenomenon that we wish to investigate is that uniformly over all sets $S\in \SND$, the behavior of $H(S)$ and $D(S)$ is nearly optimal when $d$ is large.  
It is already known \cite{devroye09} that $C_H(1) = C_H(2) = 6.21956\dots$ and that $C_H(d) < C_H(1)$ for $d \ge 3$, thus showing that hyperplane search trees outperform random binary search trees or $k$-d trees for all dimensions, with the improvement being strict when $d \ge 3$.  The present note makes this more precise, and shows that in fact, $\lim_{d\to\infty}C_H(d) = 1$.  Thus, by pushing up $d$, we can ensure almost perfectly balanced trees almost all the time, as we have the trivial lower bound $H(S) \ge \log_2 n$.
We also derive expressions and bounds on $C_H(d)$ and $C_D(d)$ as we proceed.

It is natural to go beyond this result and ask how quickly random hyperplane search trees become perfectly balanced as $d$ increases.  We find that the constants corresponding to height and average depth decay at different rates.  The main contribution of this paper is proving asymptotically optimal bounds for these rates, stated in the following two theorems:
\begin{theorem}\label{thm:height}
\begin{enumerate}
\item $C_H(d) = 1 + \BIGO{d^{-1/2}}$.
\item This bound is asymptotically optimal since there exists a function $g_H(d) = 1+\BIGOMEGA{d^{-1/2}}$ such that
$$
\lim_{n\to\infty} \, \max_{S\in\SND} \, \PROB{~\frac{H(S)}{\log_2 n} \,\geq\, g_H(d)~} ~ = ~ 1~.
$$
\end{enumerate}
\end{theorem}
\begin{theorem}\label{thm:depth}
\begin{enumerate}
\item $C_D(d) = 1 + \BIGO{d^{-1}}$.
\item This bound is asymptotically optimal since there exists a function $g_D(d) = 1+\BIGOMEGA{d^{-1}}$ such that
$$
\lim_{n\to\infty} \, \max_{S\in\SND} \, \PROB{~\frac{D(S)}{\log_2 n} \,\geq\, g_D(d)~} ~ = ~ 1~.
$$
\end{enumerate}
\end{theorem}


\subsection{Outline}

In Section \ref{sec:momentcurve} we examine random hyperplane search trees built on \emph{moment curve point sets}.  These point sets are conjectured to yield the most unbalanced random hyperplane splits.  We discuss their connection with median-of-$(2t+1)$ trees and give simple, closed-form asymptotic lower bounds for the constants governing the height and depth of these trees.  These provide the tightness parts of Theorems \ref{thm:height} and \ref{thm:depth}.

In Section \ref{sec:auxiliary}, we propose several simple lemmas that are good enough to
provide tight asymptotics for the height.  In Section \ref{sec:height} we consider the height of dominated trees and prove Theorem \ref{thm:height} using our simple lemmas from Section \ref{sec:auxiliary}.

In Section \ref{sec:dominated} we introduce two lemmas providing simple and powerful bounds for the analysis of random split trees.  The first lemma bounds the logarithmic moment of a class of random variables that often arise in the analysis of random split trees.  The second lemma bounds the depth of a random split tree using a \emph{dominating split variable}.  We apply these lemmas to obtain an almost-tight depth bound using our simple geometric lemmas from the previous section.

Finally, in Section \ref{sec:kfacets} we introduce a stronger balance lemma proved by Wagner \cite{wagner06} and restate it in the language of this paper.  Using this stronger balance lemma, we prove an asymptotically tight depth bound.

\section{Moment curve point sets and median-of-$(2t+1)$ trees}\label{sec:momentcurve}

The tightness parts of Theorems \ref{thm:height} and \ref{thm:depth}
can be shown for the moment curve data $X = \{x_1,\dots,x_n\}$ where
$$
x_i = (i, i^2, \ldots , i^d)
$$
The points on the moment curve are parametrically ordered,
and thus we can order them by first coordinate and refer
to the points by their index between $1$ and $n$.

Analysis of random hyperplane splits on such point sets is quite clean.  Choose $d$ integers uniformly at random.  This yields $d+1$ (possibly empty) intervals into which
the other points fall. Number the intervals.
One side of the hyperplane corresponds to all odd-numbered intervals,
and the other side to the even-numbered ones.
If the intervals catch $N_1, N_2, \ldots, N_{d+1}$ points
(with sum $n-d$, of course), then one subtree of the root
has size $N_1 + N_3 + \cdots$ data points, and the other one
$N_2 + N_4 + \cdots$ data points.
A statistically equivalent description yielding the same
interval sizes uses $n$ i.i.d.\ uniform $[0,1]$ random
variables $U_1, \ldots U_n$. Use $U_1, \ldots , U_d$ to
define the $d+1$ uniform spacings of $[0,1]$,
which we call $S_1, \ldots, S_{d+1}$.
Then ``throw'' the remaining $n-d$ points into the intervals.
The cardinalities are distributed as $(N_1,\ldots,N_{d+1})$,
and are multinomial $(n-d, S_1, \ldots, S_{d+1})$.
The size of the odd side of the hyperplane thus
is distributed as a sum of multinomial components---it is
binomial $(n-d, S_1 + S_3 + \cdots )$.
It is well-known that uniform spacings are identically distributed
and that their distribution is permutation-invariant (see, e.g., Pyke, 1965).
Thus, the odd side of the hyperplane is of size distributed as
a binomial $(n-d, S_1 + S_2 + \cdots + S_{(d+1)/2})$ if $d$ is odd
and as a binomial $(n-d, S_1 + S_2 + \cdots + S_{(d+2)/2})$ if $d$ is even.
But $S_1 + S_2 + \cdots + S_k$ is distributed as a beta $(k, d+1-k)$
random variable.
Thus, the root split for the moment curve data yields a left subtree
that is binomial $(n-d, \hbox{\rm beta}((d+1)/2, (d+1)/2) )$
when $d$ is odd. For $d$ even, we have with equal probability
a binomial $(n-d, \hbox{\rm beta}((d+2)/2, d/2) )$ and 
a binomial $(n-d, \hbox{\rm beta}(d/2 , (d+2)/2 )$.
One can verify that this is in turn distributed as 
a binomial $(n-d, \hbox{\rm beta}(d/2 , d/2) )$.  If we wish to consider the \emph{fraction} of points on one side of the hyperplane as $n\to\infty$ for some fixed $d$, the expression is even cleaner---the limiting distribution is simply beta$(\ceil{d/2},\ceil{d/2})$ (see, \eg, Devroye \cite[Lem. 2]{devroye90m-ary} or King \cite[\S 5.2]{king2010thesis}).

\begin{figure}
\centering
\includegraphics[width=0.48\textwidth]{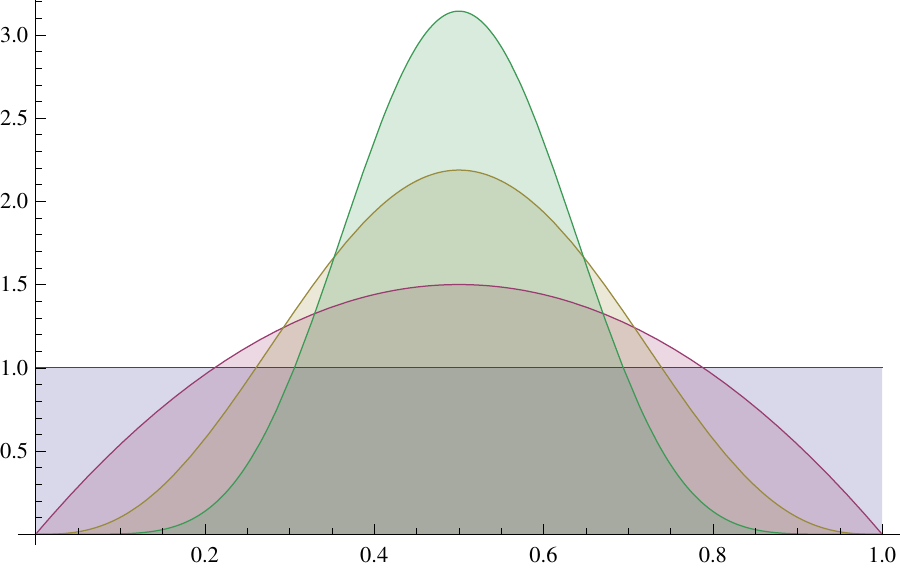}\hfill
\includegraphics[width=0.48\textwidth]{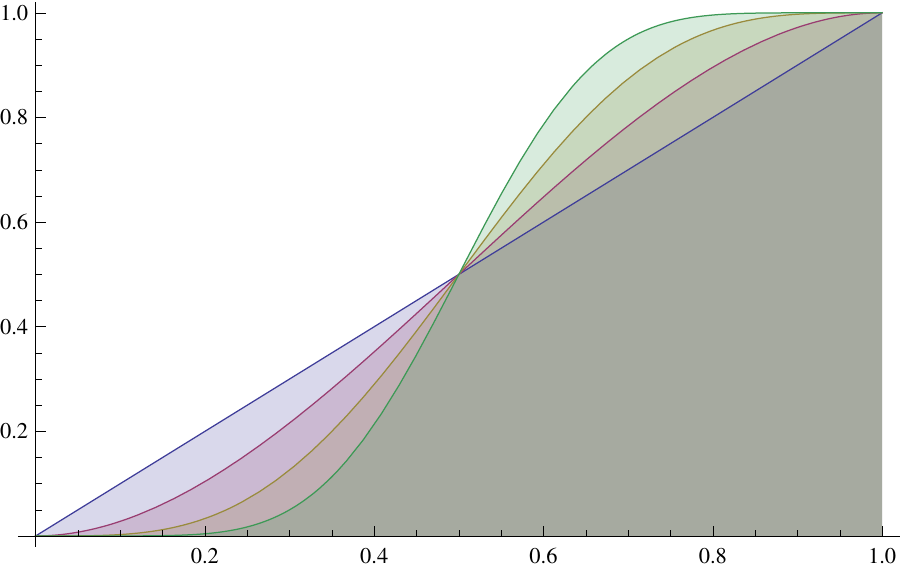}
\caption{\label{fig:beta}The PDFs (left) and CDFs (right) of the limiting split distribution beta$(\ceil{d/2},\ceil{d/2})$  for $d=1,3,7,15$.  The distribution is uniform when $d=1$ and becomes more tightly concentrated around $1/2$ as $d$ increases.}
\end{figure}

This tree is indistinguishable from the
fringe-balanced, or median-of-$(2t+1)$ search tree which has been studied
quite extensively in the data structure literature.
First suggested by Bell \cite{bell65} and Walker and Wood \cite{walker76},
it is a binary tree constructed on real-valued data.
It samples $2t+1$ data points uniformly without replacement
from the $n$ data points, where $t$ is an integer.
It then chooses the middle (median) element, and partitions
the remaining data points into two sets by using this median point.
Assuming without loss of generality that the data points are
$U_1,\ldots, U_n$, as above, we see that the leftmost set
in the split is precisely binomial $(n-(2t+1), \hbox{\rm beta} (t+1, t+1) )$.
Depending upon the implementation, the $2t$ unused pivot points
can  also be reused in the partition, thus inflating the subtree
sizes by $t$ each. For first-order asymptotics, this is an irrelevant
choice. If they are not reused, then the median-of-$(2t+1)$ tree
is distributed as the hyperplane search tree for the moment curve
if we take odd $d = 2t+1$. As we observed above, the moment curve
hyperplane search tree for $d=2t+2$ is nearly identical, i.e., at
least the beta components are of identical parameters.
Thus, we will only consider odd $d$.
 
\begin{figure}
\centering
\includegraphics{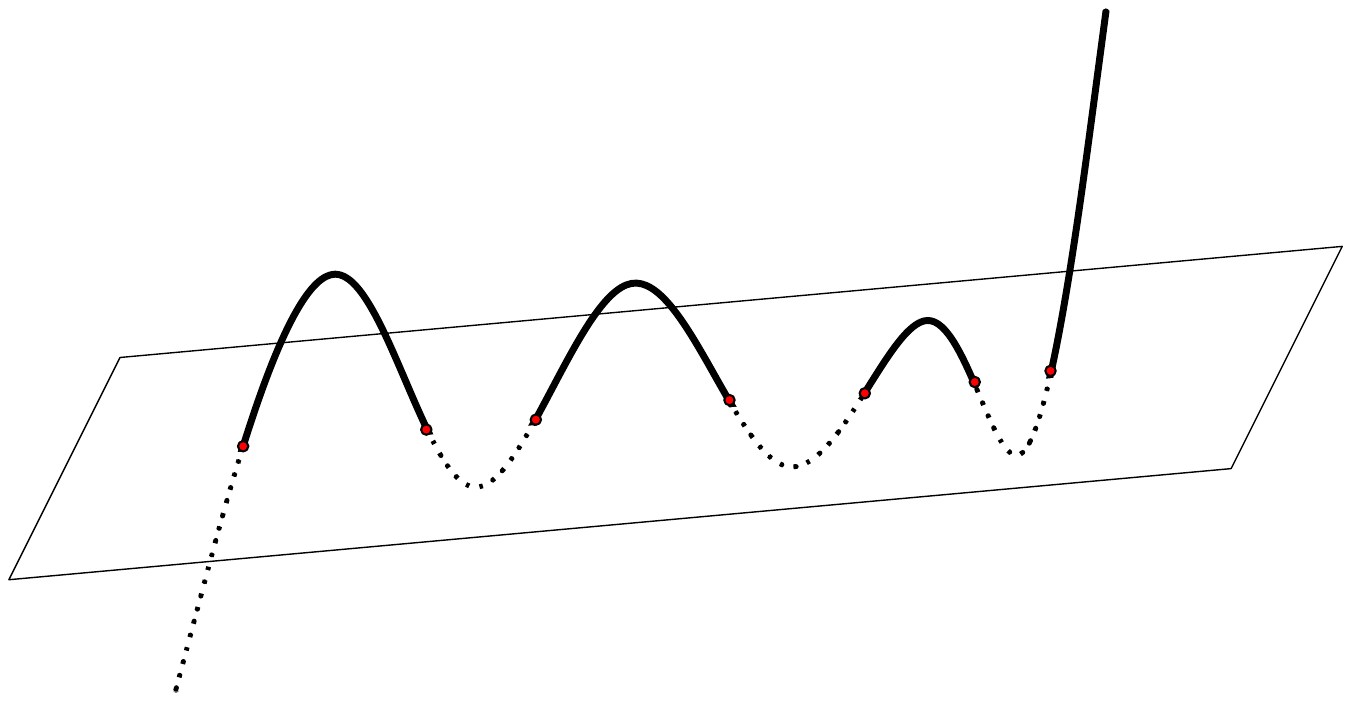}
\caption{A conceptual visualization of splits caused by choosing random points on the moment curve data set. The data are alternating above and below the hyperplane through the chosen points.}
\end{figure}

\subsection{A lower bound for the depth}

The depth $D_n$ has been studied 
by the theory of Markov processes or urn models in a series of papers, notably 
by Poblete and Munro \cite{poblete85}, Aldous et al.\ \cite{aldous88}. See also Gonnet
and Baeza-Yates \cite[p.~109]{gonnet91} and Devroye \cite{devroye99}, where
a central limit theorem for $D_n$ can be found.
Poblete and Munro \cite{poblete85} showed that 
$$
\frac{D_n}{\log n} \to \frac{1}{\sum_{i=t+1}^{2t+1} \frac{1}{i+1} } \isdef \Lambda (t)
~\hbox{\rm in probability}.
$$
Here we give a clean lower bound for $\Lambda (t)$ that proves the second half of Theorem \ref{thm:depth}.
\begin{proposition}\label{prop:moment_curve_depth}
For all $t$ sufficiently large,
$$\frac{D_n}{\log n} \geq \frac{1}{\log 2} + \frac{\log (3/2)}{4t}.$$
\end{proposition}
\begin{proof}
We know that the $n$th partial sum in the harmonic series is
$$
\log n + \gamma + {1 \over 2n} - {1 \over 12n^2} + { 1 \over 120n^4} + \cdots,
$$
where $\gamma = 0.57721\ldots$ is the Euler-Mascheroni constant.
Thus, the limit of $D_n / \log n$ is
$$
{ 1 \over \log \left( {2t+2 \over t+1} \right) + {1 \over 4t+2} - { 1 \over 2t+2 } + \BIGO{ { 1 \over t^2 } } }
= { 1 \over \log ( 2 ) - {1 \over 4t} + \BIGO{ 1 \over t^2 } }
= { 1 \over \log ( 2 ) } +   {\log (2) \over 4t} + \BIGO{ 1 \over t^2 }.
$$
For odd $d$, this is $1/\log (2) + \log(2)/(2d) + \BIGO{1/d^2}$, thus proving the second half of Theorem \ref{thm:depth}.
\end{proof}

The law of large numbers for the height
is due to Devroye (1993). We have
$$
{H_n \over \log n } \to C(t) \text{~in probability},
$$
where $C(t)$ is the unique solution $c$ greater than $\Lambda (t)$ of the equation
$$
\lambda (c)  - c \sum_{i=t+1}^{2t+1} \log \left( 1 + {\lambda (c) \over i } \right) + c \log 2 = 0,
$$
and $\lambda (c)$ is defined by the implicit equation
$$
{1 \over c} = \sum_{i=t+1}^{2t+1} { 1 \over \lambda + i }.
$$
We have $C(t) \to 1/\log 2$ as $t \to \infty$. 
A table of numerical values is given in Devroye (1993).
For example, for the moment curve in dimensions 1 and 2, the behavior is as
for random binary search trees:
$D_n / \log n \to 2$ in probability
and $H_n / \log n \to 4.31107\ldots$ in
probability.
In dimension 3, we have a beta$(2,2)$ parameter in the split vector,
and obtain
$D_n / \log n \to 12/7$ in probability
and $H_n / \log n \to 3.19257\ldots$ in
probability. For $d=2$, this is optimal as shown by
Devroye, King and McDiarmid \cite{devroye09}. For $d=3$, the moment curve
yields indeed the worst point configuration, thanks to a result
of Welzl \cite{welzl01}.

To show the last part of Theorem \ref{thm:height}, we need to show
that $C(t) \ge 1/\log (2) + c/\sqrt{t}$ for some positive constant $c$ and all $t$ large enough. 

\begin{proposition}\label{prop:moment_curve_height}
For any constant $c < \sqrt{\log (2)}$ and all $t$ sufficiently large,
$$C(t) \geq \frac{1}{\log 2} + \frac{c}{\sqrt{t}}.$$
\end{proposition}
\begin{proof}

We reparametrize with respect to $t$ as follows.
Define $\lambda = \alpha \sqrt{t}$ and $1/c = \log (2) - \beta/\sqrt{t}$.
We plug this back into the definitions of $\lambda$ and $c$,
and note that it suffices to show that as $t \to \infty$, $\beta$ tends
to a positive constant.
First note that
\begin{align*}
\sum_{i=t+1}^{2t+1} { 1 \over \lambda + i }
&= \log \left( { \lambda + 2t + 1 \over \lambda + t } \right) + \BIGO{\frac{1}{t}} \cr
&= \log (2) + \log \left( { \lambda + 2t + 1 \over 2\lambda + 2t } \right) + \BIGO{\frac{1}{t}} \cr
&= \log (2) + \log \left( 1 - { \lambda + 1 \over 2\lambda + 2t } \right) + \BIGO{\frac{1}{t}} \cr
&= \log (2) - { \lambda + 1 \over 2\lambda + 2t } + \BIGO{\frac{1}{t}} \cr
&= \log (2) - { \alpha \over 2 \sqrt{t} } + \BIGO{\frac{1}{t}}. \cr
\end{align*}
Thus, $| \alpha / 2  - \beta | = \BIGO{1/\sqrt{t}}$.
The second equation relating $\lambda$ to $c$ can be rewritten
$$
{\lambda \over c \log (2) } + 1 - {1 \over \log 2} \sum_{i=t+1}^{2t+1} \log \left( 1 + {\lambda \over i } \right)  = 0.
$$
With the reparametrization, and dividing by $\sqrt{t}$, this yields
$$
\alpha \left(1 - {\beta \over \log (2) \sqrt{t}} \right) + {1 \over \sqrt{t}} - {1 \over \log (2) \sqrt{t} } \sum_{i=t+1}^{2t+1} \log \left( 1 + {\alpha \sqrt{t} \over i } \right)  = 0.
$$
Assuming $\alpha$ remains bounded, the last term is 
\begin{align*}
&{1 \over \log (2) \sqrt{t} } \sum_{i=t+1}^{2t+1} {\alpha \sqrt{t} \over i } - {1 \over \log (2) \sqrt{t} } \sum_{i=t+1}^{2t+1} { \alpha^2 t \over 2i^2 } + \BIGO{\frac{1}{t^2}} \cr
&\qquad= {\alpha \over \log (2) } \left( \log (2) - {1 \over 4t} + \BIGO{\frac{1}{t^2}} \right) - {\alpha^2 \sqrt{t} \over \log (2) } \sum_{i=t+1}^{2t+1} { 1 \over 2i^2 } + \BIGO{\frac{1}{t^2}} \cr
&\qquad= \alpha - {\alpha \over 4t \log 2 } - {\alpha^2 \over 4 \sqrt{t} \log (2) } + \BIGO{\frac{1}{t^{3/2}}}. \cr
\end{align*}
Putting things together, our equation becomes
$$
- {\alpha \beta \over \log (2) \sqrt{t}}  + {1 \over \sqrt{t}}  + {\alpha^2 \over 4 \sqrt{t} \log (2) } + \BIGO{\frac{1}{t}} = 0.
$$
The main term is $o(1/\sqrt{t})$ if
$$
\alpha \beta - \log (2)  - {\alpha^2 \over 4}  = o(1).
$$
This happens if $\alpha \to \sqrt{ \log (16) }$,
and thus $\beta \to \sqrt{\log (2)}$.
\end{proof}

\section{Simple balance lemmas}\label{sec:auxiliary}

We need some preliminary results on polytopes.  In particular, for a polytope $P$ of $\RR^d$ with $n$ vertices, let $f_k (P)$ denote the number of $k$-faces.  A special place is occupied by $\Cnd$, the cyclic polytope in $\RR^d$ having $n$ vertices.  As a canonical example of such a polytope we can consider the convex hull of the points $\big\{ (t^1, t^2, \ldots, t^d): t = 1,2,\ldots,n \big\}$.

McMullen's Upper Bound Theorem (McMullen, 1970; McMullen and Shephard, 1971) \cite{mcmullen70, mcmullen71convex}
states that for all $1 \le k \le d-1$,
$$
\max_P f_k (P) = f_k (\Cnd).
$$
For more on this, and alternate proofs, see, \eg, Mulmuley (1994) \cite{mulmuley94},
Ziegler (1995) \cite{ziegler95}, or Kalai (1997) \cite{kalai97}.
Exact expression are well-known for $f_k (\Cnd)$. The one that is of most 
interest to us is
$$
f_{d-1} (\Cnd) = \binom{n -  \left\lfloor \frac{d+1}{2} \right\rfloor}{n-d} 
            + \binom{n -  \left\lfloor \frac{d+2}{2} \right\rfloor}{n-d} .
$$
This counts the number of full $(d-1)$-dimensional faces (\ie, facets) of $\Cnd$.
For example, when $n=d+2$, one can readily verify these formulas:
$$
f_{d-1} (\Cnd) = 
\begin{cases}
\frac{(d+2)^2}{4} & \text{when $d$ is even ,} \\
\frac{(d+1)(d+3)}{4} & \text{when $d$ is odd.} \\
\end{cases}
$$
We also note (see, \eg, Gr\"unbaum (2003) \cite{grunbaum03}) that if we are given $n = d+2$ points in convex position and in general position, then their convex hull $P$ is a simplicial polytope with $n = d+2$ vertices.  Such a polytope must be combinatorially equivalent to $\Cnd$: in particular, $f_k (P) = f_k (\Cnd)$ for all $1 \le k \le d-1$.  The following lemma is fundamental.

\begin{lemma}[The small balance lemma]\label{lem:small_balance}
Consider $d+2$ points in general position in $\RR^d$.  Let $H$ be the hyperplane through $d$ of them, chosen uniformly at random. Let $A$ be the event that the two remaining points are on the same side of $H$. Then
$$
\PROB{ A } \le \frac{1}{2} +  \frac{1}{2(d+1)}.
$$
\end{lemma}
\begin{proof}
If the points are in convex position as well, then by remarks from
the previous section (combinatorial equivalence with  the cyclic polytope
$\Cnd$ for $n=d+2$), we see that

\begin{eqnarray*}
\PROB{ A } 
&=& { \frac{f_{d-1} (\Cnd)}{\binom{d+2}{d}}} \\
&=& 
\begin{cases} 
  { \frac{d+2}{2(d+1)}} & \text{when $d$ is even,} \\
  { \frac{d+3}{2(d+2)}} & \text{when $d$ is odd} \\
\end{cases} \\
&=& 
\begin{cases} 
 \frac{1}{2} +  \frac{1}{2(d+1) } & \text{when $d$ is even,} \\
 \frac{1}{2} +  \frac{1}{2(d+2) } & \text{when $d$ is odd.} \\
\end{cases} \\
\end{eqnarray*}
If the points are not in convex position, then $d+1$ of them
form a simplex, and one point is strictly inside it.
We see that $\PROB{ A }$ now equals the probability that
if we choose $d$ points uniformly at random, we fail to pick that
interior point. Thus,
$$
\PROB{ A } = { 2 \over d+2 }.
$$
Combining cases, we note that 
$$
\PROB{ A } \le {1 \over 2} +  { 1 \over 2(d+1) }.
$$
\end{proof}

Even though it is very simple, we already note that hyperplane splits in sets as small as $d+2$ are roughly balanced for large $d$.  Next, we derive an inequality for hyperplane splits for general $n > d$.

\begin{lemma}[The balance lemma]\label{lem:balance}
Consider $n \ge d+1$ points in general position in $\RR^d$.  Let $H$ be the hyperplane through $d$ of them, chosen uniformly at random. This splits the remaining $n-d$ points into two sets, $S$ and $S'$.  Let $N = |S| \xi + |S'| (1-\xi)$, where $\xi \in \{ 0, 1 \}$ is $\BERNOULLI{1/2}$.  Then, for $x \ge 0$,
$$
\PROB{ N \ge {n-d \over 2} + x } 
~\le~ { {n-d \over 4} + {(n-d)^2 \over 4(d+1)} \over {n-d \over 4} + {(n-d)^2 \over 4(d+1)} + x^2 }~.
$$
\end{lemma}
\begin{proof}
When $n=d+1$, then the upper bound is more than $1/2$ when $x \le (n-d)/2$. For $x > (n-d)/2$, the left-hand-side is zero. So assume $n \ge d+2$.  Let $H$ denote the random set of $d$ points (instead of the hyperplane that passes through them). For $x_i \not\in H$, let $A(x_i, H)$ denote the event that $x_i$ is at the same side of $H$ as the origin (or any other arbitrary fixed point in general position with the others). Set $Y(X_i, H) = 1$ if $A(x_i, H)$ is true and $Y(X_i, H) = -1$ otherwise. Thus,
$$
N 
~=~ 
\xi \sum_{i: x_i \not\in H }  { Y(X_i, H) + 1 \over 2 } + (1-\xi) { -Y(X_i, H) + 1 \over 2 }
~=~ 
{n-d \over 2} + (\xi-1/2)   \sum_{i: x_i \not\in H } Y(X_i, H).
$$
Thus, $\EXP{ N } = (n-d)/2$, and
\begin{eqnarray*}
\VAR{N}
&=& {1 \over 4} \EXP{ \left( \sum_{i: x_i \not\in H } Y(X_i, H) \right)^2 } \\
&=& {1 \over 4} \EXP{ \sum_{i: x_i \not\in H } Y^2(X_i, H) } + { 1 \over 4 } \EXP{ \sum_{i\not=j: x_i \not\in H, x_j \not\in H } Y(x_i, H) Y(x_j, H) } \cr
&=& {n-d \over 4} + { (n-d)(n-d-1) \over 4 } \EXP{ Y(x_Z, H) Y(x_W, H) }~,
\end{eqnarray*}
where $W,Z$ are randomly drawn without replacement from $\{ x_1,\ldots, x_n\} \setminus H$.  Continuing, 
\begin{eqnarray*}
\VAR{N}
&=& {n-d \over 4} + { (n-d)(n-d-1) \over 4 } \left(\Big. 2\cdot\PROB{\big. x_Z, x_W \hbox{\rm ~are on same side of } H } - 1 \right)~.
\end{eqnarray*}
Now, after first conditioning on the set $H \cup \{x_Z , x_W\}$, which has cardinality $d+2$, Lemma \ref{lem:small_balance} gives us
\begin{eqnarray*}
\VAR{N}
&\le& {n-d \over 4} + { (n-d)(n-d-1) \over 4(d+1) } \cr
&\le& {n-d \over 4} + { (n-d)^2 \over 4(d+1) }~.
\end{eqnarray*}
By the Chebyshev-Cantelli inequality, we have
$$
\PROB{ N \ge {n-d \over 2} + x } \le { \VAR{ N } \over \VAR{ N } + x^2 }.
$$
Plugging in the upper bound on $\VAR{ N }$ gives the result. 
\end{proof}

It is convenient to have a simpler bound than that of Lemma \ref{lem:balance}, in which the sample size $n$ is removed.  For example, this suffices for our main result:

\begin{lemma}[The simplified balance lemma]\label{lem:simplified}
With notation from Lemma \ref{lem:balance}, for $x > 1/2$ we have
$$
\PROB{ {N \over n}  \ge  x  } \le \min \left(\, {1 \over 2}\, , \,{ 1 \over 1 + 4 (d+1) (x-1/2)^2}\, \right).
$$
\end{lemma}
\begin{proof}
The $1/2$ bound follows from the symmetry in the definition of $N$.  We begin by formally replacing $x$ in Lemma \ref{lem:balance} by $n(x-1/2)+d/2 = (n-d)(x-1/2) +dx $, and noting that this is $\ge (n-d)(x-1/2)$:
\begin{eqnarray*}
\PROB{ N \ge nx } 
&\le& { {n-d \over 4} + {(n-d)^2 \over 4(d+1)} \over {n-d \over 4} + {(n-d)^2 \over 4(d+1)} + (n-d)^2 (x-1/2)^2} \cr
&=& { {1 \over 4} + {n-d \over 4(d+1)} \over {1 \over 4} + {n-d \over 4(d+1)} + (n-d) (x-1/2)^2} \cr
&=& { {n+1 \over 4(d+1)} \over {n+1 \over 4(d+1)} + (n-d) (x-1/2)^2} \cr
&=& { n+1 \over n+1 + 4 (d+1) (n-d) (x-1/2)^2} \cr
&\le& { n-d \over n-d + 4 (d+1) (n-d) (x-1/2)^2} \cr
&=& { 1 \over 1 + 4 (d+1) (x-1/2)^2}. \cr
\end{eqnarray*}
\end{proof}

\section{Dominated binary trees}\label{sec:dominated}

We consider the following general set-up.
A tree with $1 \le n \le d$ data points is not split and consists
of  a single node, the root, which ``holds'' all data points.
if $n > d$, the root is split in some manner, resulting
in left and right subtree sizes $L, R$, satisfying
$L+R = n-d$, and $(L/n,R/n)$ stochastically dominated by
$(Z,1-Z)$, where $Z\in [0,1]$ is a given random variable
symmetric about $1/2$. By stochastic domination, we mean that
$$
\PROB{ \max (L/n, R/n) \ge x} \le \PROB{ \max (Z, 1-Z) n\ge x}, x \ge 0.
$$
From Marshall and Olkin (1979), we recall that for any convex
function $\psi$,
$$
\EXP{ \psi (L/n) + \psi (R/n)} \le \EXP{ \psi (Z)  + \psi (1-Z)} = 2 \EXP{ \psi (Z)}.
$$
This splitting property is recursively applied to each subtree,
and given a subtree size (like $L$), and given the data points that
are in the subtree, we require the inequality uniformly over all
point sets.
To save space, we say that we have a tree
\underbar{dominated by $Z$}.
Let us give two examples.

\begin{example}
In the random binary search tree, where $d=1$, we know that
$L \inlaw R \inlaw \lfloor nU \rfloor$, where $U$ is uniform $[0,1]$.
It is trivial to show that $(L/n,R/n)$
 is stochastically dominated by $(U, 1-U)$. 
 Thus the tree is dominated by $Z=U$.
\hfill$\square$
\end{example}
\begin{example}\label{ex:hyperplane}
The hyperplane search tree in $\RR^d$.
Let the largest of the two subtrees of the root have size $N$.
By the union bound and Lemma \ref{lem:simplified}, for $x > 1/2$,
$$
\PROB{ {N \over n}  \ge  x } \le { 2 \over 1 + 4 (d+1) (x-1/2)^2} < { 1 \over 2 (d+1) (x-1/2)^2}.
$$
Let $W \in [1/2,1]$ be a random variable with distribution function given by
$$
\PROB{ W \ge  x} 
= 
\begin{cases} 
 { 1 \over 2 (d+1) (x-1/2)^2} & \text{if $x \in \left[1/2 + \sqrt{1/2(d+1)} , 1 \right]$,} \\
       0 & \text{if $x > 1$.} \\ 
\end{cases}
$$
That is, $W$ is supported on $\left[1/2 + \sqrt{1/2(d+1)} , 1 \right]$ and has
an atom of weight $2/(d+1)$ at $1$.
It takes a moment to verify that
$$
W \inlaw {1 \over 2} + \min \left( {1 \over 2} , \sqrt{ 1 \over 2(d+1) U } \right).
$$
Thus, the $Z$ in the preceding discussion can be taken
$$ 
Z = {1 \over 2} + \sigma \min \left( {1 \over 2} , \sqrt{ 1 \over 2(d+1) U } \right),
$$
where $\sigma \in \{ -1, +1 \}$ is a random equiprobable sign.
We will see a stronger domination result for hyperplane search trees further on.\hfill$\square$
\end{example}


\section{Height of dominated trees and a proof of Theorem \ref{thm:height}}\label{sec:height}

We recall the notion of a tree dominated by $Z$, and
define 
$$
Z^* = \max ( Z , 1-Z ) .
$$

\begin{lemma}\label{lem:height}
For constant $\gamma>0$, if 
$$
\inf_{\lambda > 0} e^\lambda \left( \EXP{ 2{Z^*}^\lambda } \right)^\gamma < 1,
$$
then
$$
\lim_{n\to\infty}\PROB{ H_n > \ceil{\gamma\log n} } = 0.
$$
\end{lemma}
\begin{proof}
Let $t=\ceil{\gamma\log n}$.
By domination, we know that both subtrees of the root are
stochastically not larger than $nZ^*$.
By repeating this observation as we descend away from the root
following any path of length $t$, we deduce that
the size of the subtree at that node is stochastically not
larger than
$$
n \prod_{i=1}^t Z_i^*,
$$
where $Z_1^*, Z_2^*, \ldots$ is an i.i.d.\ sequence distributed
as $Z^*$.
Therefore, by the union bound, and the fact that we have a binary tree,
using Markov's inequality, and a constant $\lambda > 0$,
\begin{align*}
\PROB{ H_n > t } 
&\le 2^t \PROB{ n \prod_{i=1}^t Z_i^* > d } \cr
&\le 2^t \left( { n \over d+1 } \right)^{\lambda} \EXP{ \left( \prod_{i=1}^t Z_i^* \right)^\lambda } \cr
&=  \left( { n \over d+1 } \right)^{\lambda} \left( 2 \EXP{ {Z^*}^\lambda } \right)^t. \cr
\end{align*}
The upper bound is not more than
$$
\left[  e^\lambda \left( \EXP{ 2{Z^*}^\lambda } \right)^\gamma \right]^{\log n},
$$
which tends to zero if
$$
e^\lambda \left( \EXP{ 2{Z^*}^\lambda } \right)^\gamma < 1.
$$
\end{proof}

\begin{proof}[Proof of Theorem \ref{thm:height}]
Let us take
$$
Z^* \inlaw \min \left( {1 \over 2} + a \sqrt{E+b}, 1 \right) \le {1 \over 2} + a \sqrt{E+b},
$$
where $E$ is standard exponential and $a, b > 0$.
Then 
$$
\EXP{ (2Z^*)^\lambda } 
\le \EXP{ \exp ( 2a\lambda \sqrt{E+b} ) }.
$$
Choose $\lambda = 1/(2a)$, and define $\rho = \EXP{ \exp ( \sqrt{E+b} ) }$.
Then
$$
e^\lambda \left( \EXP{ 2{Z^*}^\lambda } \right)^\gamma
\le \exp\left( \lambda + \gamma ( \log (2\rho) - \lambda \log 2 ) \right)
< 1
$$
provided
$$
\gamma  
> { \lambda \over \lambda \log 2 - \log (2 \rho) }
= { 1 \over \log 2 - 2a \log (2 \rho) }.
$$
In particular, if $a = \Theta (1/\sqrt{d})$, then
this, along with Lemma \ref{lem:height}, would imply Theorem \ref{thm:height}.
But Example \ref{ex:hyperplane} implies that a hyperplane search tree is dominated
by precisely such a $Z^*$, with $a = 1 / \sqrt{2d}$ and $b = \log 8$.
\end{proof}

\section{Logarithmic moments and depth of dominated trees}\label{sec:depth}

The depth of a random node in a tree dominated by $Z$
is determined by the \underbar{logarithmic} \underbar{moment}
$$
\mu 
= 2 \EXP{ Z \log (1/Z)}
= \EXP{ W \log (1/W)  + (1-W) \log (1/(1-W))}
= \EXP{ Y }
$$
where $Y$ is a random variable defined as follows:
$$
Y 
= 
\begin{cases}
  \log \left( { 1 \over W } \right) & \text{with probability $W$,} \\
  \log \left( { 1 \over 1-W } \right) & \text{with probability $1-W$.} \\
\end{cases}
$$
Note that since $x \log x$ is bounded on $[0,1]$, 
$\mu \ge 0$ is bounded. Also, $\mu = 0$ if and only if
$Z \in \{ 0, 1\}$, i.e., $Z$ is Bernoulli $(1/2)$ (recalling that $Z$
is symmetric).  We first provide two useful general lemmas for
computing a bound on the logarithmic moment and obtaining
a one-sided law of large numbers for general $Z$.

\begin{lemma}\label{lem:symmetric}
For a random variable $Z = 1/2 + \sigma V$, where $V$ is
$[0,1/2]$-valued, and $\sigma \in \{ -1, +1 \}$ is a random independent equiprobable sign,
$$
\mu \ge \log 2 -  \alpha \EXP{ V^2 },
$$
where $\alpha = 2 (1+ \sqrt{\log 8})^2 < 19$.
\end{lemma}
\begin{proof}
Note that
$$
-\mu 
= \EXP{ \Big.\left( \textstyle \frac{1}{2} + V \right) \log \left( \textstyle \frac{1}{2} + V \right) + \left( \textstyle \frac{1}{2} - V \right) \log \left( \textstyle \frac{1}{2} - V \right) \Big.}
= -\log 2 + \EXP{ f(2V) },
$$
where
$$
f(v) \isdef { (1+v) \log (1+v) + (1-v) \log (1-v)  \over 2}, 0 \le v \le 1.
$$
We check that $f(0) = f'(0) = 0$, $f''(v) = 1/(1-v)^2 \ge 0$,
so $f$ is convex and increasing to $f(1) = \log 2$.
On $[0,b]$, with $b < 1$, we have
$$
f(v) \le {1 \over (1-b)^2} \times {v^2 \over 2}.
$$
On $[b,1]$, we have $f(v) \le \log (2) \le (v/b)^2 \log 2$.
Combining this and choosing $b = \sqrt{\log 8} / (1+ \sqrt{\log 8})$, we see that
$$
f(v) \le  {1 \over 4} \, \alpha v^2, 0 \le v \le 1.
$$
\end{proof}

\begin{lemma}\label{lem:domination}
In a random binary tree dominated by $Z$,  having logarithmic
moment $\mu > 0$, we have for every $\epsilon > 0$,
$$
\lim_{n \to \infty} \PROB{ {D_n \over \log n}  \ge {1 \over \mu} + \epsilon } = 0.
$$
\end{lemma}
\begin{proof}
Let us begin with a small observation. Let $\lambda > 0$ be a parameter
and let $X \le 0$ be a nonpositive random variable. Then
$$
\lim_{\lambda \downarrow 0} \EXP{ { e^{\lambda X} - 1 \over \lambda } } = \EXP{ X }.
$$
This is best seen by noting that $(e^{\lambda x} - 1)/\lambda \ge x$, which provides a lower bound.
Since $(e^{\lambda X} - 1)/\lambda \le 0$, we have by Fatou's lemma,
$$
\limsup_{\lambda \downarrow 0} \EXP { { e^{\lambda X} - 1 \over \lambda } } 
\le 
\EXP{ \limsup_{\lambda \downarrow 0}{ e^{\lambda X} - 1 \over \lambda } }
= \EXP{ X }.
$$
Thus, as $\lambda \downarrow 0$,
$$
\varphi (\lambda) \isdef \EXP{ e^{-\lambda Y} } = 1 - \lambda \EXP{ Y } + o(\lambda) = 1 - \lambda \mu + o(\lambda).
$$
Let us show by induction on the integers $t$ that for all $\lambda \ge 0$, $n \ge 1$,
$$
\PROB{ D_n \ge t } \le n^\lambda ( \varphi (\lambda) )^t, t \ge 0.
$$
Assuming this for a moment, then we have with $t = \lceil (1/\mu + \epsilon) \log n \rceil$,
\begin{align*}
\PROB{ D_n \ge t } 
&\le n^\lambda ( 1 - \lambda \mu + o(\lambda) )^t \\
&\le \left[ e^\lambda ( 1 - \lambda \mu + o(\lambda) )^{1/\mu + \epsilon} \right]^{\log n} \\
&= \left[  1 - \lambda \mu \epsilon + o(\lambda) \right]^{\log n} \\
&= o(1) 
\end{align*}
if we choose $\lambda > 0$ small enough but fixed.
This would complete the proof.

For the proof by induction, note that for $t=0$, the inequality is trivial.
So, we consider a general $t > 0$. Then denoting by $X(L)$ and $X(R)$ the subsets of data points that end up in the left and right subtrees of the root, and by $D_L$ and $D_R$ the depths of random nodes (relative to their subtree roots) of data points randomly selected from $X(L)$ and $X(R)$,
respectively,
then
\begin{align*}
\PROB{ D_n \ge t}
&\le \EXP{ {L \over n} \PROB{ D_L \ge t-1 | X(L)} + {R \over n} \PROB{ D_R \ge t-1 | X(R)}  } \\
&\le \EXP{ {L \over n} L^\lambda ( \varphi (\lambda) )^{t-1} + {R \over n} R^\lambda ( \varphi (\lambda) )^{t-1} } \\
&\le n^\lambda ( \varphi (\lambda) )^{t-1} \EXP{ \left( {L \over n} \right)^{\lambda + 1} + \left( {R \over n} \right)^{\lambda + 1} } \\
&\le n^\lambda ( \varphi (\lambda) )^{t-1} \EXP{  Z^{\lambda + 1} + (1-Z)^{\lambda + 1} } \\
		&= n^\lambda ( \varphi (\lambda) )^{t-1} \EXP{ e^{-\lambda Y} } \\
		&= n^\lambda ( \varphi (\lambda) )^t.
\end{align*}
\end{proof}

The inequalities of Section \ref{sec:auxiliary} are powerful enough to obtain a depth bound that is almost asymptotically tight.  Our proof is of independent interest since it uses only Section \ref{sec:auxiliary}, the previous two lemmas, and simple textbook arguments.

\begin{proposition}\label{prop:depth}
Consider a  hyperplane search tree
for a collection of points $x_1, x_2,  \ldots , x_n \in \RR^d$ that are in general position.
For fixed $d$, there exists a constant $C(d)$ such that for all $\epsilon > 0$,
$$
\lim_{n \to \infty} \sup_{x_1,\ldots,x_n \in \RR^d} 
\PROB{ D_n \ge (C(d)+\epsilon) \log_2 n } = 0.
$$
Furthermore, as $d \to \infty$,
$$
C(d) = 1 + \BIGO{\frac{\log (d)}{d}}.
$$
\end{proposition}
\begin{proof}
Observe that in the definition of $\mu$, we can take
$$ 
Z = \frac{1}{2} + \sigma V,
$$
where
$$
V \isdef \min \left( \frac{1}{2} , \sqrt{ 1 \over 2(d+1) U } \right).
$$
By Lemma \ref{lem:symmetric}, for this $Z$,
\begin{align*}
\mu 
&\ge  \log 2 -  \alpha \EXP{ V^2 } \\
&=  \log 2 -  \alpha \EXP{ \min \left( {1 \over 4} , { 1 \over 2(d+1) U } \right) } \\
&=  \log 2 - {\alpha \over 4}  \EXP{ \min \left( 1 , { 2 \over (d+1) U } \right) } \\
&=  \log 2 - { \alpha \over 2(d+1)} \left( 1 + \log \left( { d+1 \over 2 } \right) \right).
\end{align*}
Combining Lemma \ref{lem:domination} with this then completes the proof.
\end{proof}

\section{Stronger bounds on $({\leq}k)$-facets and a proof of Theorem \ref{thm:depth}}\label{sec:kfacets}

Analysis of random hyperplane splits is directly related to the problem of counting $k$-facets in discrete geometry.  For a set of $n$ points in general position in $\RR^d$, a subset of $d$ points, along with an orientation, defines an oriented hyperplane with an associated positive open halfspace.  If this halfspace contains exactly $k$ of the remaining $n-d$ points, we say that the oriented set of $d$ points is a $k$-facet.  Thus each subset of $d$ points defines, for some $0\leq k \leq \lfloor{(n-d)/2} \rfloor$, a $k$-facet with one orientation and an $(n-d-k)$-facet with the other orientation.  A $({\leq}k)$-facet is simply a $j$-facet for some $j\leq k$.  Knowing the probability mass function of $N$ is equivalent to knowing the number of $({\leq}k)$-facets for every $0\leq k < \lfloor{(n-d)/2} \rfloor$.  For a thorough treatment of $k$-facets we direct the reader to Wagner's 2008 survey \cite{wagner08}.

A significant open conjecture in discrete geometry is the Spherical Generalized Upper Bound Conjecture, or {\sc sgubc}.  Forms of this conjecture were proposed independently by Eckhoff \cite{eckhoff93}, Linhart \cite{linhart94}, and Welzl \cite{welzl01}.  Wagner \cite[Conjecture 1.2]{wagner06} proposes the conjecture in full generality.  Here we state a slightly weaker form of the conjecture in the language of this paper, which would be implied by {\sc sgubc}.

\begin{conjecture}
For a set of $n\geq d$ points in general position in $\RR^d$, define random variable $N$ as the number of points on the larger side of a random hyperplane split.  Define $N^*$ as the analogous random variable for the larger side of a random hyperplane split for the moment curve data in $\RR^d$.
Then, for any $x > (n-d)/2$, 
$$
\PROB{ N \geq x } \leq \PROB{ N^* \geq x }.
$$ 
\end{conjecture}

The {\sc sgubc} conjecture is trivially true for $d=1$. 
For $d=2$, it was proved by Peck \cite{peck85} and Alon and Gy\H{o}ri \cite{alon86}.  Welzl \cite{welzl01} proved Conjecture 1 for $d=3$.  
Inequalities for the far right tail of $N$ were
obtained by Clarkson and Shor \cite{clarkson89}. 
For general $d\geq 1$, Wagner \cite{wagner06} proved a relaxed form of a conjecture closely related to the {\sc sgubc}
 that implies Lemma \ref{lem:wagner} below.

\begin{lemma}[Wagner]\label{lem:wagner}
With $N$ and $N^*$ defined as above, for any $x > (n-d)/2$, 
$$
\PROB{ N \geq x } \leq 4 \cdot \PROB{ N^* \geq x }.
$$
\end{lemma}

Clarkson and Shor \cite{clarkson89} proved a somewhat similar result many years earlier, but their bound only holds for the extreme tail of the distribution, \ie, as $x$ approaches 1.  It is therefore insufficient for our purposes.  Wagner's bound, on the other hand, is valid for the entire range of $x$ that concerns us.  In order to exploit this result using the machinery of the previous section, we must first bound $\PROB{ N^* \geq x }$ in an appropriate manner.

\begin{lemma}
For a set of $n$ distinct points on the moment curve in $\RR^d$, we have
$$
\PROB{ {N^* \over n}  \ge  x }
 \le \exp\left(-  2d\left(x - 1/2 \right)^2 \right), x \ge {1 \over 2}.
$$
\end{lemma}
\begin{proof}
After first fixing $x \ge 1/2$, for the sake of analysis we introduce random variables 
$$
B \inlaw \hbox{\rm beta} \left( {\lceil{d/2}\rceil, \lceil{d/2}\rceil} \right),
$$
and $\xi_{d,x}$ with a binomial $(d,x)$ distribution. 
It is known (see Devroye \cite{devroye90m-ary}) that $\PROB{ N^* / n  \ge  y } \leq \PROB{ \max (B, 1-B)  \ge  y }$ for all $y\geq 0$, \ie,
$N^*/n$ is stochastically dominated by $\max (B, 1-B)$. 
It is also known that the c.d.f.'s of $B$ and $\xi_{d,x}$ are duals: for $x \in (0,1)$,
$$
\PROB{ B  \ge  x }
=
\begin{cases}
\PROB{ \xi_{d,x} \leq (d-1)/2},& \text{if $d$ is odd,} \cr
{1 \over 2} \, \PROB{ \xi_{d,x} \leq (d-2)/2 }+ {1 \over 2} \PROB{ \xi_{d,x} \leq d/2 },& \text{if $d$ is even.} \cr
\end{cases}
$$
Thus,
\begin{align*}
\PROB{ {N^* \over n}  \ge  x } 
&\leq \PROB{ \max(B, 1-B)  \ge  x }  \cr
&= 2 \PROB{ B  \ge  x }  \cr
&\leq 2 \PROB{ \xi_{d,x}  \le  \lfloor{d/2}\rfloor } \cr
&\leq \exp{\left(-{ 2(dx - \lfloor{d/2}\rfloor)^2 \over d} \right)}  \quad\hbox{\rm (by Hoeffing's inequality \cite{hoeffding1963probability})} \cr
&\leq \exp\left(- { 2d^2(x - 1/2)^2 \over d} \right)  \cr
&= \exp{\left(-2d\left(x - 1/2 \right)^2\right)} ,
\end{align*}concluding the proof.
\end{proof}

\begin{proof}[Proof of Theorem \ref{thm:depth}]
Let $E$ be a standard exponential random variable, 
let $\sigma \in \{ -1,+1\}$ be a random equiprobable sign,
and define
$$
V = \min \left( {1 \over 2} , \sqrt{ E + \log 4 \over 2d } \right),
Z = { 1 \over 2}  + \sigma V.
$$
Note that for $1 \ge x \ge 1/2$,
\begin{align*}
\PROB{ {N \over n } \ge x }
&\le 4 \exp\left(-  2d\left(x - 1/2 \right)^2 \right) \cr
&= \PROB{ {1 \over 2} + V \ge x } \cr
&= \PROB{ \max (Z, 1-Z) \ge x },
\end{align*}
and thus the hyperplane search tree is dominated by this
$Z$.

The logarithmic moment $\mu$ of $Z$ is easily bounded
using Lemma \ref{lem:symmetric}.  With $\alpha = 2 (1 + \log 8)^2$ we have
\begin{align*}
\mu 
&\geq \log 2-\alpha \EXP{V^2} \cr
&= \log 2-\alpha \EXP{ \min\left({1\over 4},{E + \log 4 \over 2d}\right) } \cr
& = \log 2-{\alpha \over 4d}  \EXP{ \min\left(d,2E + 2\log 4 \right) } \cr
& \geq \log 2- {\alpha \over 4d}  \EXP{ 2E + 2\log 4 } \cr
& = \log 2- {\alpha (1 + \log 4) \over 2d}. \cr
\end{align*}
This implies that $1/\mu = (1+\BIGO{1/d})\log 2$ and therefore Theorem \ref{thm:depth} follows from Lemma \ref{lem:domination}. Thus, the sharper estimates
for domination that flow from Wagner's inequality give the optimal
rate of convergence with respect to $d$.
\end{proof}


\bibliographystyle{plain}
\bibliography{hyperplane_tree_2}

\begin{thebibliography}{10}

\bibitem{aldous88}
D.~Aldous, B.~Flannery, and J.L. Palacios.
\newblock {Two applications of urn processes: The fringe analysis of search
  trees and the simulation of quasi-stationary distributions of Markov chains.}
\newblock {\em Probability in the Engineering and Informational Sciences},
  2(3):293--307, 1988.

\bibitem{alon86}
N.~Alon and E.~Gy{\H{o}}ri.
\newblock {The number of small semispaces of a finite set of points in the
  plane}.
\newblock {\em Journal of Combinatorial Theory, Series A}, 41(1):154--157,
  1986.

\bibitem{bell65}
C.J. Bell.
\newblock {\em {An Investigation into the Principles of the Classification and
  Analysis of Data on an Automatic Digital Computer}}.
\newblock PhD thesis, Leeds University, 1965.

\bibitem{clarkson89}
K.L. Clarkson and P.W. Shor.
\newblock {Applications of random sampling in computational geometry, II}.
\newblock {\em Discrete and Computational Geometry}, 4(1):387--421, 1989.

\bibitem{devroye86}
L.~Devroye.
\newblock A note on the height of binary search trees.
\newblock {\em Journal of the ACM}, 33:489--498, 1986.

\bibitem{devroye87}
L.~Devroye.
\newblock Branching processes in the analysis of the heights of trees.
\newblock {\em Acta Informatica}, 24:277--298, 1987.

\bibitem{devroye88}
L.~Devroye.
\newblock Applications of the theory of records in the study of random trees.
\newblock {\em Acta Informatica}, 26:123--130, 1988.

\bibitem{devroye90m-ary}
L.~Devroye.
\newblock On the height of random {$m$}-ary search trees.
\newblock {\em Random Struct. Algorithms}, 1(2):191--204, 1990.

\bibitem{devroye99}
L.~Devroye.
\newblock Universal limit laws for depths in random trees.
\newblock {\em SIAM Journal on Computing}, 28:409--432, 1999.

\bibitem{devroye96}
L.~Devroye, L.~Gy{\"o}rfi, and G.~Lugosi.
\newblock {\em A Probabilistic Theory of Pattern Recognition}.
\newblock Springer Verlag, New York, 1996.

\bibitem{devroye09}
L.~Devroye, J.~King, and C.~McDiarmid.
\newblock Random hyperplane search trees.
\newblock {\em SIAM Journal on Computing}, 38(6):2411--2425, 2009.

\bibitem{eckhoff93}
J.~Eckhoff.
\newblock {Helly, Radon, and Carath{\'e}odory type theorems}.
\newblock {\em Handbook of Convex Geometry}, pages 389--448, 1993.

\bibitem{edelsbrunner83}
H.~Edelsbrunner and J.~van Leeuwen.
\newblock Multidimensional data structures and algorithms: a bibliography.
\newblock Technical report, Technische Universit\"at Graz, 1983.

\bibitem{gonnet91}
G.H. Gonnet and R.~Baeza-Yates.
\newblock {\em Handbook of algorithms and data structures}.
\newblock Addison-Wesley Longman Publishing Co., Inc. Boston, MA, USA, 1991.

\bibitem{grunbaum03}
B.~Gr{\"u}nbaum.
\newblock {\em {Convex polytopes}}.
\newblock Springer Verlag, 2003.

\bibitem{haussler87}
D.~Haussler and E.~Welzl.
\newblock Epsilon-nets and simplex range queries.
\newblock {\em Discrete {\&} Computational Geometry}, 2:127--151, 1987.

\bibitem{hoeffding1963probability}
W.~Hoeffding.
\newblock {Probability inequalities for sums of bounded random variables}.
\newblock {\em Journal of the American Statistical Association},
  58(301):13--30, 1963.

\bibitem{kalai97}
G.~Kalai.
\newblock {Linear programming, the simplex algorithm and simple polytopes}.
\newblock {\em Mathematical Programming}, 79(1):217--233, 1997.

\bibitem{king2010thesis}
J.~King.
\newblock {\em Guarding Problems and Geometric Split Trees}.
\newblock PhD thesis, McGill University, 2010.

\bibitem{knuth73}
D.E. Knuth.
\newblock {\em The art of computer programming. Vol. 3: Sorting and searching}.
\newblock Addison-Wesley, 1973.

\bibitem{linhart94}
J.~Linhart.
\newblock {The Upper Bound Conjecture for arrangements of halfspaces}.
\newblock {\em Contributions to Algebra and Geometry}, 35(1):29--35, 1994.

\bibitem{lynch65}
W.C. Lynch.
\newblock {More combinatorial problems on certain trees}.
\newblock {\em Computer Journal}, 7:299--302, 1965.

\bibitem{mahmoud92}
H.~M. Mahmoud.
\newblock {\em Evolution of Random Search Trees}.
\newblock John Wiley, New York, 1992.

\bibitem{mcmullen70}
P.~McMullen.
\newblock {The maximum numbers of faces of a convex polytope}.
\newblock {\em Mathematika}, 17:179--184, 1970.

\bibitem{mcmullen71convex}
P.~McMullen and G.C. Shephard.
\newblock {\em {Convex polytopes and the upper bound conjecture}}.
\newblock Cambridge University Press, 1971.

\bibitem{mizoguchi77}
R.~Mizoguchi, M.~Kizawa, and M.~Shimura.
\newblock Piecewise linear discriminant functions in pattern recognition.
\newblock {\em Systems, Computers, and Control}, 8:114--121, 1977.

\bibitem{mulmuley94}
K.~Mulmuley.
\newblock {\em Computational geometry: an introduction through randomized
  algorithms}.
\newblock Prentice Hall, 1994.

\bibitem{peck85}
G.W. Peck.
\newblock On {$k$}-sets in the plane.
\newblock {\em Discrete Mathematics}, 56(1):73--74, 1985.

\bibitem{pittel84}
B.~Pittel.
\newblock On growing random binary trees.
\newblock {\em Journal of Mathematical Analysis and Applications},
  103:461--480, 1984.

\bibitem{poblete85}
P.V. Poblete and J.I. Munro.
\newblock The analysis of a fringe heuristic for binary search trees.
\newblock {\em Journal of Algorithms}, 6(3):336--350, 1985.

\bibitem{robson79}
J.M. Robson.
\newblock The height of binary search trees.
\newblock {\em Australian Computer Journal}, 11(4):151--153, 1979.

\bibitem{wagner06}
U.~Wagner.
\newblock On a geometric generalization of the upper bound theorem.
\newblock In {\em Proceedings of the 47th Annual IEEE Symposium on Foundations
  of Computer Science}, pages 635--645, 2006.

\bibitem{wagner08}
U.~Wagner.
\newblock {$k$}-sets and {$k$}-facets.
\newblock {\em Contemporary Mathematics}, 453:443, 2008.

\bibitem{walker76}
A.~Walker and D.~Wood.
\newblock Locally balanced binary trees.
\newblock {\em The Computer Journal}, 19(4):322, 1976.

\bibitem{welzl01}
E.~Welzl.
\newblock Entering and leaving {$j$}-facets.
\newblock {\em Discrete and Computational Geometry}, 25(3):351--364, 2001.

\bibitem{you76}
K.~C. You and K.~S. Fu.
\newblock An approach to the design of a linear binary tree classifier.
\newblock In {\em Proceedings of the Symposium of Machine Processing of
  Remotely Sensed Data}, volume Technical Report 3A-10, Purdue University,
  1976.

\bibitem{ziegler95}
G.M. Ziegler.
\newblock {\em {Lectures on Polytopes}}.
\newblock Graduate Texts in Mathematics. Springer, New York, 1995.

\end{thebibliography}

\end{document}